%% file: ffpf.tex
\documentclass[letterpaper,10pt,twoside]{article}
\usepackage[utf8]{inputenc}
\usepackage[english]{babel}
\usepackage[T1]{fontenc}
\usepackage{lmodern}
\usepackage{amsmath,amssymb,amsthm,mathtools} 
\usepackage[ps,arrow,matrix]{xy} 

\newcommand{\im}{\mathrm{i}}
\newcommand{\N}{\mathbb{N}}
\newcommand{\Z}{\mathbb{Z}}
\newcommand{\R}{\mathbb{R}}
\newcommand{\C}{\mathbb{C}}
\newcommand{\defeq}{\coloneqq}
\newcommand{\tens}{\otimes}
\DeclareMathOperator{\ctens}{\hat{\otimes}}
\newcommand{\ds}{\circ}
\newcommand{\cH}{\mathcal{H}}
\newcommand{\toi}{\hookrightarrow}

\newcommand{\sig}[1]{[#1]}
\newcommand{\fdg}[1]{|#1|}
\newcommand{\sts}{\mathcal{M}}
\newcommand{\lhs}{\langle\langle}
\newcommand{\rhs}{\rangle\rangle}
\newcommand{\cD}{\mathcal{D}}
\newcommand{\cDd}{\mathcal{D}^{\ds}}
\newcommand{\cDr}{\mathcal{D}^{\R}}
\newcommand{\cDrd}{\mathcal{D}^{\R\ds}}
\newcommand{\cDp}{\mathcal{D}^{+}}
\newcommand{\cDpd}{\mathcal{D}^{+\ds}}
\newcommand{\cF}{\mathcal{F}}
\newcommand{\cFr}{\mathcal{F}^{\R}}
\newcommand{\hol}{\mathrm{h}}
\newcommand{\ahol}{\overline{\mathrm{h}}}
\newcommand{\one}{\mathbf{1}}

\theoremstyle{definition}
\newtheorem{dfn}{Definition}[section]

\theoremstyle{plain}
\newtheorem{lem}[dfn]{Lemma}

\newtheorem{thm}[dfn]{Theorem}

\allowdisplaybreaks

\begin{document}

\input{titlepage}

\input{introduction}

\input{motivation}

\input{twistedpf}

\input{axrealpos}

\input{fockfeyn}

\input{addcomplex}

\input{outlook}

\subsection*{Acknowledgments}

This work was supported in part by UNAM--DGAPA--PAPIIT through project grant IN100212.

\appendix

\input{sts}

\input{classax}

\bibliographystyle{stdnodoi} 
\bibliography{stdrefsb}
\end{document}

%% file: titlepage.tex
\begin{titlepage}
\title{\textbf{Towards state locality in quantum\\ field theory: free fermions}}
\author{Robert Oeckl\footnote{email: robert@matmor.unam.mx}\\ \\
Centro de Ciencias Matemáticas,\\
Universidad Nacional Autónoma de México,\\
C.P.~58190, Morelia, Michoacán, Mexico}
\date{UNAM-CCM-2013-2\\ 18 July 2013\\ 10 December 2015 (v2)\\ 12 September 2016 (v3)}

\maketitle

\vspace{\stretch{1}}

\begin{abstract}
\input{abstract}
\end{abstract}

\vspace{\stretch{1}}
\end{titlepage}

%% file: abstract.tex
Hilbert spaces of states can be constructed in standard quantum field theory only for infinitely extended spacelike hypersurfaces, precluding a more local notion of state. In fact, the Reeh-Schlieder Theorem prohibits the localization of states on pieces of hypersurfaces in the standard formalism. From the point of view of geometric quantization the problem lies in the non-locality of the complex structures associated to hypersurfaces in standard quantization. We show that using a weakened version of the positive formalism puts this problem into a new perspective. This is a local TQFT type formalism based on super-operators and mixed state spaces rather than on amplitudes and pure state spaces as the one of Atiyah-Segal. In particular, we show that in the case of purely fermionic degrees of freedom the complex structure can be dispensed with when the notion of state is suitably generalized. These generalized states do localize on compact hypersurfaces with boundaries. For the simplest case of free fermionic fields we embed this in a rigorous and functorial quantization scheme yielding a local description of the quantum theory. Crucially, no classical data is needed beyond the structures evident from a Lagrangian setting. When the classical data is augmented with complex structures on hypersurfaces, the quantum data correspondingly augment to the full positive formalism. This scheme is applicable to field theory in curved spacetime, but also to field theories without metric background.

%% file: introduction.tex
\section{Introduction}

In the quantum theory of systems with finitely many degrees of freedom there is a simple relation between the state space of a composite system and those of its components: The Hilbert space of states of the composite system is the tensor product of the Hilbert spaces of states of the individual subsystems. This is no longer true at the more fundamental level of quantum field theory. If we understand by ``subsystem'' the physics happening in a certain part of space or spacetime, there is no Hilbert space of states that we can assign to this ``subsystem'' in such a way that the physics outside of this part of space or spacetime is completely excluded. In particular, we cannot piece together the complete Hilbert space of states as simply the tensor product of Hilbert spaces associated to different such ``subsystems''. This has been known for a long time and we shall refer to this as the \emph{state locality problem} in quantum field theory.

 One way to see this is via the Reeh-Schlieder Theorem \cite{ReSc:unitequiv}. This says essentially that applying operators representing observables in an arbitrarily small spacetime region to the vacuum state yields a dense subspace of the whole Hilbert space of states. But if we had a Hilbert space of states associated to a part of space or spacetime then operators in that part should only act on this ``partial'' Hilbert space. So the Reeh-Schlieder Theorem would imply that this ``partial'' Hilbert space must already be the whole Hilbert space.\footnote{We are somewhat simplifying the discussion here in assuming also a ``localizable'' notion of vacuum. One could imagine more complicated set-ups where the Reeh-Schlieder Theorem would not straightforwardly apply. But then the physical interpretation of the ``partial'' Hilbert space would also be less straightforward.}

It is clear that in order to solve the state locality problem we cannot stick to Hilbert spaces and their tensor products. In fact, if we are willing to considerably modify the notion of state, \emph{algebraic quantum field theory} \cite{Haa:lqp} may be seen to provide a solution. There, the fundamental mathematical objects are observable algebras associated to spacetime regions. States are secondary objects that arise as functionals on these algebras. Thus, states may trivially be viewed as localized on the same spacetime regions. Moreover, a global state induces a localized state by restriction of the corresponding functional from the global algebra of observables to its subalgebra localized in the spacetime region of interest.

We are interested, however, in a stronger and more fundamental notion of state locality, one that implies a composability analogous to the tensor product of Hilbert spaces. In fact, the notion we shall consider applies already at the level of merely topological spacetime, before even the introduction of a metric background.

The framework we shall use to make precise and address the state locality problem is the \emph{general boundary formulation} (GBF) of quantum theory \cite{Oe:gbqft}. The GBF is being developed with the aim of providing a formulation of quantum theory that is manifestly spacetime local and makes sense even in the absence of separate notions of time, space or metric (although relying on a weak notion of spacetime). This is to be contrasted with the standard formulation of quantum theory (in terms of a Hilbert space of states and observables as operators on it) which is not local in space and relies heavily on a predetermined notion of time.

The GBF consists of two components, exactly as the standard formulation of quantum theory. The first specifies what a quantum theory ``is''. This takes the form of an axiomatic formalism describing the mathematical ingredients of a quantum theory and their interrelations. The second specifies how predictions for measurement outcomes are extracted from the formalism. This takes the form of rules to calculate probabilities and expectation values.

There has been a long-standing effort, led principally by G.~Segal, on understanding the structure of quantum field theory in terms of the mathematical framework of \emph{topological quantum field theory} (TQFT) \cite{Ati:tqft}. (See \cite{Seg:cftproc,Seg:cftdef} for conformal field theory and \cite{Seg:localholom,Seg:fklectures} more generally.) A version of this, which we shall refer to as the \emph{amplitude formalism}, has served so far precisely as the first component of the GBF. (This was first proposed in \cite{Oe:boundary}, for the latest version see Appendix~A in \cite{Oe:dmf}.) Roughly, in the amplitude formalism a Krein space of generalized states is associated to each hypersurface in spacetime. Moreover, to each region of spacetime there is associated a linear complex valued \emph{amplitude map} on the Krein space of the boundary. For any region there is moreover one such map per observable that it may contain, called \emph{observable map}.

The second component of the GBF consists of the rules to extract probabilities and expectation values. (The probability rule was first formulated in \cite{Oe:gbqft} and that for expectation values in \cite{Oe:obsgbf}.)

Suppose we are given a quantum field theory on a globally hyperbolic spacetime in terms of the objects of the amplitude formalism (and satisfying its axioms). It is then straightforward \cite{Oe:holomorphic,Oe:feynobs,Oe:freefermi} to extract the ingredients of the standard formulation. Notably, ``the'' Hilbert space arises as associated to an (arbitrarily chosen) Cauchy hypersurface. Moreover, the usual transition probabilities and expectation values are recovered exactly for spacetime regions bounded by pairs of Cauchy hypersurfaces.

In this sense, the GBF reproduces exactly the results of the standard formulation. Indeed, it is a \emph{reformulation} and not a \emph{modification} of quantum theory.
Conversely, if we were given a quantum field theory in terms of a Hilbert space and spacetime localized field operators on a globally hyperbolic manifold, we can straightforwardly construct the objects of the amplitude formalism and satisfy its axioms.\footnote{The case of non-relativistic quantum theory is even simpler as spacetime may then be taken to be the real line representing time.} That is, we can do so \emph{if} we restrict ourselves to Cauchy hypersurfaces and spacetime regions that are bounded by pairs of these.

So, in a rather trivial way a restricted version of the GBF is essentially equivalent to the standard formulation. However, the main motivation for setting up the GBF arises from its promise to describe physics in a manifestly spacetime local way. What we need is for the amplitude formalism to include spacetime regions that are compact and arbitrarily small, and moreover regions that permit the gluing together to form larger regions. This in turn means that we need hypersurfaces that arise as boundaries of such regions or arise as interfaces for the gluing. In particular, we need hypersurfaces with boundary that we can decompose into pieces that are hypersurfaces with boundary. This is where we meet again the state locality problem. It does not seem possible in general to assign Hilbert (or Krein) spaces to such bounded hypersurfaces and have the decomposition correspond to a tensor product decomposition.

Thus it seems that the GBF with the amplitude formalism must fail when we try to exploit fully its locality in quantum field theory. Before proceeding let us emphasize, however, that there is also a rich intermediate regime, where the state locality problem does not occur, but where the GBF is strictly more powerful than the standard formulation. This is the regime where we include timelike hypersurfaces (without boundary) and regions bounded by these. In this context the GBF has been successfully implemented in a range of examples, leading to interesting results, see e.g.\ \cite{Oe:timelike,Oe:kgtl,CoOe:smatrixgbf,CDO:adsproc,CoRa:unruh,CoRa:qftrindler}. Moreover, the state locality problem does not occur in two dimensions, where Yang-Mills theory was solved with hypersurfaces with boundary \cite{Oe:2dqym}. In a Riemannian spacetime the problem is also less severe \cite{CoOe:2deucl}.

Recently, a novel \emph{positive formalism} was proposed to take the role of the amplitude formalism in the GBF \cite{Oe:dmf}. The main motivation for its introduction was operationalism, but as we shall see it is useful in addressing state locality as well. From a physical perspective the transition from the amplitude formalism to the positive formalism is somewhat analogous to the transition from a pure state to a mixed state formalism in the standard formulation. From a mathematical perspective, the positive formalism is also an incarnation of TQFT, but one where realness and positivity play a distinguished role. Moreover, the path from the amplitude formalism to the positive formalism is a kind of ``taking the modulus square'' of a TQFT.

The present paper is precisely about implementing state locality via the positive formalism. For simplicity we restrict ourselves to linear field theory. Also, we do not consider observables in this paper. The main results will be limited to fermionic field theory.

In Section~\ref{sec:sloc} we review aspects of the state locality problem in the context of the GBF and outline the approach to its solution followed in the rest of this paper. In Section~\ref{sec:twistedpf} we review the relation between the amplitude formalism and the positive formalism and describe the ``modulus squared'' construction. Axioms of a ``real formalism'' are described in Section~\ref{sec:axrealpos}. This is a stripped down version of the positive formalism. In Section~\ref{sec:fockfeyn} we describe a new quantization scheme applicable to classical free fermionic field theories that targets this real formalism. This is related to the quantization scheme of \cite{Oe:freefermi} in Section~\ref{sec:addcomplex}. A discussion of the obtained results and outlook on further questions is presented in Section~\ref{sec:outlook}. Appendix~\ref{sec:sts} contains a brief description of the notion of spacetime in the GBF. Appendix~\ref{sec:classax} contains axioms for classical free field theory.

%% file: motivation.tex
\section{The state locality problem in the GBF}
\label{sec:sloc}

In this section we shall consider the state locality problem from the perspective of the quantization of free field theory in the amplitude formalism of the GBF as described in \cite{Oe:holomorphic,Oe:freefermi}.

The first ingredient is an appropriate notion of spacetime. This takes the form of a \emph{spacetime system}, i.e., a collection of hypersurfaces and regions such that the hypersurfaces can be decomposed, the regions glued, etc. These geometric objects may be manifolds in their own right or arise as submanifolds of a given manifold, see Appendix~\ref{sec:sts}. In contrast to previous versions of these axioms (see \cite{Oe:freefermi} for the latest), we explicitly allow hypersurfaces with boundaries as this is where we pretend to ``localize'' states.

Both classical and quantum field theories are encoded in terms of algebraic structures associated to the geometric objects of the spacetime system. Moreover, these structures satisfy an axiomatic system. The quantization scheme then provides a functor from a category of classical free field theories satisfying the classical axioms to a category of quantum theories satisfying the quantum axioms in the form of the \emph{amplitude formalism}.

We note that the quantization schemes of \cite{Oe:holomorphic,Oe:affine,Oe:freefermi} work perfectly well with the more general spacetime axioms as given here in Appendix~\ref{sec:sts}. That is, they are applicable without any essential modification when we include hypersurfaces with boundary. The only required change is a change of wording in the axioms (both classical and quantum), replacing the notion of ``disjoint decomposition'' of hypersurfaces with the more general one of ``decomposition''. All theorems etc.\ remain valid and the proofs generalize trivially.

There is a good reason, however, why the quantization schemes were not presented this way: The simplest quantum field theories in Minkowski spacetime satisfy the axioms for spacelike and (certain) timelike hypersurfaces that have no boundary. But they do not satisfy the axioms with hypersurfaces with boundary. The problem appears at the level of the axioms for classical field theory. This is what we focus on in the following.

The first relevant comment in this respect is that the axioms as presented in \cite{Oe:holomorphic} for free bosonic and in \cite{Oe:freefermi} for free fermionic field theory are not really purely classical. Rather they can be divided into a purely classical part and a part that already contains the seed for quantization and that we shall thus call semiclassical. By ``purely classical'' we mean here the structures that naturally arise from a Lagrangian analysis of field theory. The purely classical axioms are summarized in Appendix~\ref{sec:pclassax}. (For their motivation see \cite{Oe:holomorphic,Oe:freefermi}.)
We recall briefly how the semiclassical axioms can be obtained from these. There is a real vector space $L_{\Sigma}$ of ``germs of solutions'' associated to the hypersurface $\Sigma$ (Axiom C1). This space is additionally equipped with a non-degenerate anti-symmetric bilinear form $\omega_{\Sigma}$ (bosonic case) or a non-degenerate symmetric bilinear form $g_{\Sigma}$ (fermionic case). Consider now adding a compatible \emph{complex structure}, i.e., a linear map $J_{\Sigma}: L_{\Sigma}\to L_{\Sigma}$ such that $J_{\Sigma}^2=-\one$ and such that $J_{\Sigma}$ leaves $\omega_{\Sigma}$ or $g_{\Sigma}$ invariant. The maps $J_{\Sigma}$ also need to behave well under orientation reversal, hypersurface decomposition, gluing etc., but this is not so relevant for the moment. This complex structure makes $L_{\Sigma}$ into a \emph{complex} Hilbert space (bosonic case) or Krein space (fermionic case).\footnote{We make the simplifying assumption here that in the bosonic case the space $L_{\Sigma}$ happens to be positive-definite and complete with respect to the so defined inner product. The former property is an additional assumption and the latter may be achieved by completion.} The corresponding additional axioms are summarized in Appendix~\ref{sec:sclassax}. In total this leads to the full semiclassical axioms, see \cite[Section~4.2]{Oe:freefermi}. Concretely, the complex inner product on $L_{\Sigma}$ is,
\begin{align}
\{\phi',\phi\}_{\Sigma}& \defeq g_{\Sigma}(\phi',\phi)+2\im\omega_{\Sigma}(\phi',\phi) ,\quad\text{where}, \label{eq:complip} \\
g_{\Sigma}(\phi',\phi) & \defeq 2\omega_{\Sigma}(\phi',J_{\Sigma} \phi) \quad \text{(bosonic case)}, \label{eq:bsym} \\
\omega_{\Sigma}(\phi',\phi) & \defeq \frac{1}{2}g_{\Sigma}(J_{\Sigma} \phi',\phi) \quad \text{(fermionic case)} . \label{eq:fsympl}
\end{align}

For a free field theory the space $L_{\Sigma}$ is completely local as a symplectic vector space (in the bosonic case) or as a real Krein space (in the fermionic case). That is, if $\Sigma$ is decomposed into a union of hypersurfaces $\Sigma_1,\dots,\Sigma_n$ the space $L_{\Sigma}$ decomposes naturally as a symplectic vector space or as a real Krein space into a corresponding direct sum $L_{\Sigma_1}\oplus\cdots\oplus L_{\Sigma_n}$. For $L_{\Sigma}$ as a real vector space this is obvious since it is defined as a space of germs of solutions on $\Sigma$.\footnote{We interpret ``free field theory'' here in a strict sense as linear field theory without gauge symmetries.} But the bilinear form $\omega_{\Sigma}$ (bosonic case) or $g_{\Sigma}$ (fermionic case) is also local, because it arises as the integral over $\Sigma$ of a $d-1$-form in spacetime ($d$ being the spacetime dimension).

In contrast, the usual complex structure $J_{\Sigma}$ in field theory arises from a global distinction of ``forward'' versus ``backward'' (with respect to the hypersurface) propagating solutions. In particular, it cannot be represented as a differential operator. (It is merely a pseudodifferential operator.) In particular, $J_{\Sigma}$ does not decompose into a sum of complex structures for the components in a decomposition of $\Sigma$. This is how the state locality problem manifests itself in the present context.

There are proposals in a Riemannian context to address the problem which involve extending the component hypersurfaces with small neighborhoods and composing them with small overlaps. At the same time complex structures can be defined on the extended component hypersurfaces that are ``correct'' when viewed one the un-extended hypersurfaces. The algebraic operation associated to the composition is then a more complicated fusion tensor product \cite{Seg:fklectures}.

We shall follow a more minimalistic approach. Namely, we shall attempt to eliminate the need for the complex structure as much as possible. It is clear that a theory satisfying the semiclassical axioms gives rise to a quantum theory in terms of the amplitude formalism. But this in turn gives rise to a quantum theory in terms of the \emph{positive formalism} by ``taking the modulus square'' \cite{Oe:freefermi}. Crucially, some (operationally irrelevant) information is lost in the second step. This makes it natural to ask whether the semiclassical axioms can be weakened if we devise a quantization scheme that directly outputs to the positive formalism. The answer turns out to be affirmative, at least in the fermionic case.

We shall show that it is possible to devise a quantization scheme starting from the purely classical axioms of Appendix~\ref{sec:pclassax}, targeting a weakened version of the positive formalism. The weakening consists in removing the information about the positivity, while completely retaining real structure, see Section~\ref{sec:realax}. Crucially, we if we choose to add complex structures to a theory at the semiclassical level, the positive structure is recovered at the quantum level. Moreover, this happens in such a way that the outcome of the quantization scheme then becomes equivalent to the previous scheme that outputs to the amplitude formalism followed by ``taking the modulus square'' to arrive at the positive formalism. What is more, the addition may be partial. That is, without loss of consistency we can choose to add complex structures only on certain classes of hypersurfaces and recover positivity on these. Note in particular that the probability interpretation only acts at the level of regions and requires a positive structure only on their boundaries.

%% file: twistedpf.tex
\section{From the AF to the PF with a twist}
\label{sec:twistedpf}

We shall review in this section the ``modulus square'' operation of converting a theory in the \emph{amplitude formalism} (AF) to a theory in the \emph{positive formalism} (PF). The procedure we present here is slightly modified as compared to the one presented in \cite{Oe:freefermi}. We shall refer to this modification, which turns out to be crucial, as a \emph{twist}.

We shall make extensive use in this section of notation and results from \cite{Oe:dmf}. Recall in particular, that a (strict) \emph{Krein space} is an indefinite inner product space $V$ with a canonical orthogonal decomposition $V=V_+\oplus V_-$. Moreover, the \emph{positive part} $V_+$ is positive definite and complete and the \emph{negative part} $V_-$ is negative definite and complete. Viewed as a $\Z_2$-grading, also called \emph{signature}, we use the following notation to distinguish positive and negative parts,
\begin{equation}
 \sig{v}\defeq\begin{cases} 0\, \text{if}\, v\in V_+\\
 1\, \text{if}\, v\in V_- . \end{cases}
\end{equation}
Also define the map $I: V\to V$ to be the identity on $V_+$ and minus the identity on $V_-$.
The vector spaces of interest are also equipped with the fermionic $\Z_2$-grading, or \emph{f-grading}, that we shall denote as,
\begin{equation}
 \fdg{v}\defeq\begin{cases} 0\, \text{if $v$ is of even degree}\\
 1\, \text{if $v$ is of odd degree} .
 \end{cases}
\end{equation}
We also write $V=V_0\oplus V_1$ for the decomposition of an f-graded space.

\subsection{Complex spaces}

Thus, suppose we are given a quantum theory in terms of the amplitude formalism, i.e., in terms of the Axioms in \cite[Section~6.1]{Oe:freefermi} or equivalently in \cite[Appendix~A.2]{Oe:dmf}.
To any hypersurface $\Sigma$ we associate the Krein space $\cD_{\Sigma}$ given by the tensor product,
\begin{equation}
 \cD_{\Sigma}\defeq\cH_{\Sigma}\ctens\cH_{\overline{\Sigma}} ,
\end{equation}
completed with respect to the following inner product,
\begin{equation}
\lhs \psi'\tens \eta',\psi\tens\eta\rhs_{\Sigma}
 \defeq\langle \psi',\psi\rangle_{\Sigma} (-1)^{\fdg{\eta'}\cdot\fdg{\eta}}
  \langle \eta',\eta\rangle_{\overline{\Sigma}} .
\label{eq:iptwist}
\end{equation}
Compared to the construction presented in \cite{Oe:dmf}, the inner product here differs by a \emph{twist} in the form of a sign factor corresponding to the interchange of the components associated to $\overline{\Sigma}$. More generally, the twist consists of considering the objects associated to the orientation reversed version of the hypersurfaces in question as if arranged in the opposite order.
Since the spaces $\cH_{\Sigma}$ and $\cH_{\overline{\Sigma}}$ are both f-graded, the tensor product space $\cD_{\Sigma}$ is naturally bigraded. We use the notation
\begin{equation}
 \cD_{\Sigma,ij}=\cH_{\Sigma,i}\ctens\cH_{\overline{\Sigma},j} .
\end{equation}
$\cD_{\Sigma}$ also carries a composite f-grading with the even part given by $\cD_{\Sigma,0}=\cD_{\Sigma,00}\oplus\cD_{\Sigma,11}$ and the odd part given by $\cD_{\Sigma,1}=\cD_{\Sigma,01}\oplus\cD_{\Sigma,10}$.

Next, we define the conjugate linear adapted involution $\iota_{\Sigma}^*:\cD_{\Sigma}\to\cD_{\overline{\Sigma}}$ as in \cite{Oe:dmf} as,
\begin{equation}
  \iota_{\Sigma}^*(\psi\tens\eta)\defeq (-1)^{\fdg{\psi}\cdot\fdg{\eta}}\iota_{\Sigma}(\psi)\tens\iota_{\overline{\Sigma}}(\eta) .
\end{equation}
Note that as in \cite{Oe:dmf} this map is f-graded isometric, in spite of the twist modified inner products.
We proceed to define the map associated to decompositions of hypersurfaces. Thus, let $\Sigma=\Sigma_1\cup\Sigma_2$ be such a decomposition. Define $\tau_{\Sigma_1,\Sigma_2;\Sigma}^*:\cD_{\Sigma_1}\ctens\cD_{\Sigma_2}\to\cD_{\Sigma}$ by,
\begin{multline}
  \tau_{\Sigma_1,\Sigma_2;\Sigma}^*\left((\psi_1\tens\eta_1)\tens (\psi_2\tens\eta_2)\right)\\
\defeq (-1)^{\fdg{\eta_1}\cdot\fdg{\psi_2}+\fdg{\eta_1}\cdot\fdg{\eta_2}} \tau_{\Sigma_1,\Sigma_2;\Sigma}(\psi_1\tens\psi_2)\tens\tau_{\overline{\Sigma_1},\overline{\Sigma_2};\overline{\Sigma}}(\eta_1\tens\eta_2)
\end{multline}
Here, the difference to \cite{Oe:dmf}, i.e., the twist, lies in the sign factor corresponding to an interchange of $\eta_1$ and $\eta_2$.
The compatibility between $\iota^*$ and $\tau^*$ is easily verified to take the form,
\begin{equation}
\tau^*_{\overline{\Sigma_1},\overline{\Sigma_2};\overline{\Sigma}}
 \left(\iota^*_{\Sigma_1}(\sigma_1)\tens\iota^*_{\Sigma_2}(\sigma_2)\right)
  =(-1)^{\fdg{\sigma_1}\cdot\fdg{\sigma_2}}\iota^*_\Sigma\left(\tau^*_{\Sigma_1,\Sigma_2;\Sigma}(\sigma_1\tens\sigma_2)\right) .
\end{equation}

We associate to any region $M$ the \emph{probability map} $A_M:\cDd_{\partial M}\to\C$, composed of amplitude maps as in \cite{Oe:dmf} by,
\begin{equation}
 A_M((\psi\tens\eta))\defeq \rho_M(\psi) \rho_{\overline{M}}(\eta) .
\end{equation}
Here $\cDd_{\partial M}$ is a dense subspace of $\cD_{\partial M}$ which we do not exactly specify here. However, it contains at least the subspace $\cH^{\ds}_{\partial M}\tens \cH^{\ds}_{\partial\overline{M}}$, which is already dense.
The probability map can be seen to give rise to the inner product on the Krein spaces associated to hypersurfaces as follows. Let $\Sigma$ be a hypersurface. The boundary $\partial\hat{\Sigma}$ of the associated slice region $\hat{\Sigma}$ decomposes into the disjoint union $\partial\hat{\Sigma}=\overline{\Sigma}\cup\Sigma'$, where $\Sigma'$ denotes a second copy of $\Sigma$. We then have the following relation between probability map on $\hat{\Sigma}$ and inner product on $\Sigma$,
\begin{equation}
A_{\hat{\Sigma}}\left(\tau^*_{\overline{\Sigma},\Sigma';\partial\hat{\Sigma}}\left(\iota_{\Sigma}(\sigma)\tens\sigma'\right)\right)=\lhs\sigma,\sigma'\rhs_{\Sigma} .
\end{equation}
Also note that since we assume $\rho_{\overline{M}}(\psi)=\overline{\rho_M(\iota_{\overline{\partial M}}(\psi))}$, we have,
\begin{equation}
 A_{\overline{M}}(\sigma)=\overline{A_{M}(\iota^*_{\overline{\partial M}}(\sigma))} .
\end{equation}

\subsection{Real structure}

The map $\iota^*_{\Sigma}:\cD_{\Sigma}\to\cD_{\overline{\Sigma}}$ plays the role of a \emph{real structure} (complex conjugation) combined with a flip of the orientation of the hypersurface. To disentangle the two we introduce an identification $\cD_{\Sigma}\to\cD_{\overline{\Sigma}}$ via,
\begin{equation}
 (\psi\tens\eta)\mapsto (\eta\tens\psi) .
\end{equation}
This identification is not mandated by the structures we have described so far, but it is compatible with them, most notably with the definition of the probability map. Indeed, the identification we choose here is different (in the fermionic case) from the one used in \cite{Oe:dmf}. It differs by a twist, in the form of a missing sign factor corresponding to the interchange of $\psi$ and $\eta$. Note that the inner product does get modified when flipping orientation with this identification. With the identification implicit this relation takes the form,
\begin{equation}
\lhs \sigma',\sigma\rhs_{\Sigma}= (-1)^{\fdg{\sigma}\cdot\fdg{\sigma'}} \lhs \sigma',\sigma\rhs_{\overline{\Sigma}} .
\end{equation}

 The remaining real structure $\cD_{\Sigma}\to\cD_{\Sigma}$, which we shall write with the usual notation of a complex conjugation, takes the form,
\begin{equation}
 \overline{(\psi\tens\eta)}=(-1)^{\fdg{\psi}\cdot\fdg{\eta}} (\iota_{\overline{\Sigma}}(\eta)\tens\iota_{\Sigma}(\psi)) .
\end{equation}
The real structure commutes with the probability map via,
\begin{equation}
 A_M(\overline{\sigma})=\overline{A_M(\sigma)} .
\label{eq:realstruc}
\end{equation}
Crucially it also commutes with the inner product, i.e.,
\begin{equation}
 \lhs \overline{\sigma'},\overline{\sigma}\rhs_{\Sigma}
 = \overline{\lhs \sigma',\sigma \rhs_{\Sigma}} ,
\end{equation}
and with hypersurface decomposition, i.e.,
\begin{equation}
\tau^*_{\Sigma_1\Sigma_2;\Sigma}(\overline{\sigma_1}\tens\overline{\sigma_2})
 = \overline{\tau^*_{\Sigma_1\Sigma_2;\Sigma}(\sigma_1\tens\sigma_2)} .
\end{equation}
This is in contrast to the real structure defined in \cite{Oe:dmf}. As a consequence, when restricting the spaces $\cD_{\Sigma}$ to real subspaces $\cDr_{\Sigma}$, all relevant properties restrict as well. That is, we can write the axioms of the positive formalism here completely in terms of real vector spaces. (This was possible in \cite{Oe:dmf} only in the bosonic case.)

\subsection{Positive structure}

The real structure used in \cite{Oe:dmf} was induced from the following identification of the space $\cD_{\Sigma}$ with the space of Hilbert-Schmidt operators on $\cH_{\Sigma}$ via,
\begin{equation}
  (\psi\tens\eta)\xi \defeq \psi\, \langle I \iota_{\overline{\Sigma}}(\eta),\xi\rangle_{\Sigma} .
\end{equation}
The real structure was then given by the notion of \emph{adjoint} of an operator in the Hilbert space sense,
\begin{equation}
 \langle I \sigma^\dagger\xi,\eta\rangle_{\Sigma} =
 \langle I\xi,\sigma\eta\rangle_{\Sigma} .
\end{equation}
In contrast, the real structure here is given in terms of the operator point of view by,
\begin{equation}
 \langle I \overline{\sigma}\xi,\eta\rangle_{\Sigma} = (-1)^{\fdg{\xi}\cdot\fdg{\eta}}
 \langle I\xi,\sigma\eta\rangle_{\Sigma} .
\end{equation}
Note, however, that both notions coincide on $\cD_{\Sigma,00}$. This is significant, because for $\Sigma=\partial M$, this is the only component of $\cD_{\partial M}$ in terms of the bigrading, on which the probability map $A_M$ does not vanish. This becomes explicit in the form $A_M$ takes in the operator picture, when keeping in mind the f-gradedness of the amplitude maps,
\begin{equation}
A_M(\sigma)=\sum_{i\in I} \overline{\rho_M(\xi_i)} \rho_M(\sigma\xi_i) .
\end{equation}
Here $\{\xi_i\}_{i\in I}$ is an ON-basis of $\cH_{\Sigma}$.

The operator point of view provides in addition a \emph{positive structure}, making the real vector space of self-adjoint operators into an \emph{ordered vector space}. This is encoded through a subset $\cDp_{\Sigma}\subseteq \cD_{\Sigma}$ of positive elements in the form of positive operators. This positive structure is crucial in the probability interpretation of the positive formalism (compare \cite[Section~3.1]{Oe:dmf}). The probability interpretation relies on the probability maps $A_M$. Thus, relevant is only the set of positive elements $\cDp_{\Sigma,00}$ in the subspace $\cD_{\Sigma,00}$. In contrast to $\cDp_{\Sigma}$ as a whole, this is compatible with the real structure (\ref{eq:realstruc}) defined above.

The compatibility of the positive structure with the other structures is as follows. Most important is the already mentioned compatibility with the probability map,
\begin{equation}
 A_M(\sigma)\ge 0 \quad\text{if}\, \sigma\in \cDpd_{\Sigma} .
\end{equation}
For the inner product we get,
\begin{equation}
\lhs I^* \sigma',\sigma\rhs_{\Sigma} \ge 0\quad\text{if}\, \sigma,\sigma'\in\cDp_{\Sigma,00} ,
\end{equation}
where $I^*:\cD_{\Sigma}\to\cD_{\Sigma}$ is the identity on the positive part of $\cD_{\Sigma}$ and minus the identity on the negative part of $\cD_{\Sigma}$.
For hypersurface decompositions we have,
\begin{equation}
 \tau^*_{\Sigma_1,\Sigma_2;\Sigma}(\sigma_1\tens\sigma_2) \in \cDp_{\Sigma,00+}\quad\text{if}\, \sigma_1\in\cDp_{\Sigma_1,00+},\sigma_2 \in \cDp_{\Sigma_2,00+} .
\end{equation}

%% file: axrealpos.tex
\section{Axioms: Real plus Positive}
\label{sec:axrealpos}

We present in this section a version of the axioms of the Positive Formalism (PF) that is twist modified as compared to the axioms given in \cite{Oe:dmf} and thus fits the structures of the previous Section~\ref{sec:twistedpf}. Moreover, we split the axioms in a set that ignores the positive structure plus additional axioms to capture the positive structure. We do this with the anticipation that the axioms dealing with the positive structure might have a more limited applicability. We use the letter ``R'' for ``real'' to index the first set of axioms to express the prominence of real structure in these.

\subsection{Real axioms}
\label{sec:realax}

\begin{itemize}
\item[(R1)] Associated to each oriented hypersurface $\Sigma$ is a real separable f-graded Krein space $\cDr_\Sigma$ with inner product $\lhs\cdot,\cdot\rhs_\Sigma$. $\cDr_{\overline{\Sigma}}$ is identified with $\cDr_{\Sigma}$, but its inner product is modified via,
\begin{equation}
\lhs \sigma',\sigma\rhs_{\Sigma}= (-1)^{\fdg{\sigma}\cdot\fdg{\sigma'}} \lhs \sigma',\sigma\rhs_{\overline{\Sigma}} .
\end{equation}
\item[(R2)] Suppose the hypersurface $\Sigma$ decomposes into a union of hypersurfaces $\Sigma=\Sigma_1\cup\cdots\cup\Sigma_n$. Then, there is an isometric isomorphism of Krein spaces
  $\tau^*_{\Sigma_1,\dots,\Sigma_n;\Sigma}:\cDr_{\Sigma_1}\ctens\cdots\ctens\cDr_{\Sigma_n}\to\cDr_\Sigma$. The maps $\tau^*$ satisfy obvious associativity conditions. Moreover, in the case $n=2$ the map $(\tau^*_{\Sigma_2,\Sigma_1;\Sigma})^{-1}\circ \tau^*_{\Sigma_1,\Sigma_2;\Sigma}:\cDr_{\Sigma_1}\ctens\cDr_{\Sigma_2}\to\cDr_{\Sigma_2}\ctens\cDr_{\Sigma_1}$ is the f-graded transposition,
\begin{equation}
 \sigma_1\tens\sigma_2\mapsto (-1)^{\fdg{\sigma_1}\cdot\fdg{\sigma_2}}\sigma_2\tens\sigma_1 .
\end{equation}
Also the maps for opposite orientations are related via,
\begin{equation}
 \tau^*_{\Sigma_1,\Sigma_2;\Sigma}(\sigma_1\tens\sigma_2)=
 (-1)^{\fdg{\sigma_1}\cdot\fdg{\sigma_2}} \tau^*_{\overline{\Sigma_1},\overline{\Sigma_2};\overline{\Sigma}}(\sigma_1\tens\sigma_2) .
\end{equation}
\item[(R4)] Associated to each region $M$ is an f-graded linear map
  from a dense subspace $\cDrd_{\partial M}$ of $\cDr_{\partial M}$ to the real
  numbers, $A_M:\cDrd_{\partial M}\to\R$. This is called the
  \emph{probability map}. It is orientation independent, i.e, $A_{\overline{M}}=A_M$.
\item[(R3x)] Let $\Sigma$ be a hypersurface. The boundary $\partial\hat{\Sigma}$ of the associated slice region $\hat{\Sigma}$ decomposes into the disjoint union $\partial\hat{\Sigma}=\overline{\Sigma}\cup\Sigma'$, where $\Sigma'$ denotes a second copy of $\Sigma$. Then, $\tau^*_{\overline{\Sigma},\Sigma';\partial\hat{\Sigma}}(\cDr_{\overline{\Sigma}}\tens\cDr_{\Sigma'})\subseteq\cDrd_{\partial\hat{\Sigma}}$. Moreover, $A_{\hat{\Sigma}}\circ\tau^*_{\overline{\Sigma},\Sigma';\partial\hat{\Sigma}}:\cDr_{\overline{\Sigma}}\tens\cDr_{\Sigma'}\to\R$ restricts to the inner product $\lhs\cdot,\cdot\rhs_\Sigma:\cDr_{\Sigma}\times\cDr_{\Sigma}\to\R$.
\item[(R5a)] Let $M_1$ and $M_2$ be regions and $M\defeq M_1\cup M_2$ be their disjoint union. Then $\partial M=\partial M_1\cup \partial M_2$ is also a disjoint union and $\tau^*_{\partial M_1,\partial M_2;\partial M}(\cDrd_{\partial M_1}\tens \cDrd_{\partial M_2})\subseteq \cDrd_{\partial M}$. Moreover, for all $\sigma_1\in\cDrd_{\partial M_1}$ and $\sigma_2\in\cDrd_{\partial M_2}$,
\begin{equation}
 A_{M}\left(\tau^*_{\partial M_1,\partial M_2;\partial M}(\sigma_1\tens\sigma_2)\right)= A_{M_1}(\sigma_1) A_{M_2}(\sigma_2) .
\end{equation}
\item[(R5b)] Let $M$ be a region with its boundary decomposing as a disjoint union $\partial M=\Sigma_1\cup\Sigma\cup \overline{\Sigma'}$, where $\Sigma'$ is a copy of $\Sigma$. Let $M_1$ denote the gluing of $M$ with itself along $\Sigma,\overline{\Sigma'}$ and suppose that $M_1$ is a region. Then, $\tau^*_{\Sigma_1,\Sigma,\overline{\Sigma'};\partial M}(\sigma\tens\xi\tens\xi)\in\cDrd_{\partial M}$ for all $\sigma\in\cDrd_{\partial M_1}$ and $\xi\in\cDrd_\Sigma$. Moreover, for any orthonormal basis $\{\xi_i\}_{i\in I}$ of $\cDr_\Sigma$ in $\cDrd_{\Sigma}$, we have for all $\sigma\in\cDrd_{\partial M_1}$,
\begin{equation}
 A_{M_1}(\sigma)\cdot \hat{c}(M;\Sigma,\overline{\Sigma'})
 =\sum_{i\in I} (-1)^{\sig{\xi_i}} A_M\left(\tau^*_{\Sigma_1,\Sigma,\overline{\Sigma'};\partial M}(\sigma\tens\xi_i\tens\xi_i)\right),
\end{equation}
where $\hat{c}(M;\Sigma,\overline{\Sigma'})\in\R$ is called the (modulus square of the) \emph{gluing anomaly factor} and depends only on the geometric data.
\item[(R5b*)] Let $M$ be a region with its boundary decomposing as a disjoint union $\partial M=\Sigma_1\cup\Sigma\cup \overline{\Sigma'}$, where $\Sigma'$ is a copy of $\Sigma$. Let $M_1$ denote the gluing of $M$ with itself along $\Sigma,\overline{\Sigma'}$ and suppose that $M_1$ is a region. Then, $\tau^*_{\Sigma_1,\Sigma,\overline{\Sigma'};\partial M}(\sigma\tens\xi\tens\xi)\in\cDrd_{\partial M}$ for all $\sigma\in\cDrd_{\partial M_1}$ and $\xi\in\cDrd_\Sigma$. Moreover, there is a direct system $\{\cDr_{\Sigma,\alpha}\}_{\alpha\in A_{\Sigma}}$ of finite dimensional Krein subspaces of $\cDr_{\Sigma}$ with limit $\cDr_{\Sigma}$ and a corresponding collection $\{\hat{c}_{\alpha}(M;\Sigma,\overline{\Sigma'})\}_{\alpha\in A_{\Sigma}}$ of real numbers satisfying the following: For any orthonormal basis $\{\xi_i\}_{i\in I}$ of $\cDr_\Sigma$, we have for all $\sigma\in\cDrd_{\partial M_1}$,
\begin{multline}
 \varinjlim_{\alpha}\Bigg(
 A_{M_1}(\sigma)\cdot \hat{c}_{\alpha}(M;\Sigma,\overline{\Sigma'}) \\
 - \sum_{i\in I} (-1)^{\sig{\xi_i}} A_M\left(\tau^*_{\Sigma_1,\Sigma,\overline{\Sigma'};\partial M}(\sigma\tens P_{\alpha}\xi_i\tens P_{\alpha}\xi_i)\right)\Bigg) = 0 .
\label{eq:rencompid}
\end{multline}
Here, $P_{\alpha}$ denotes the orthogonal projector onto $\cDr_{\Sigma,\alpha}$.
\end{itemize}

We have included here two versions of the axiom concerned with the gluing of two regions along a hypersurface: (R5b) and (R5b*). Only one is supposed to be satisfied, although the validity of (R5b) implies the validity of (R5b*). While (R5b) is simpler, it has the disadvantage that the quantization scheme presented here does not automatically satisfy it. Rather an additional \emph{integrability condition} has to be imposed. This is avoided in the \emph{renormalized} version (R5b*) of the axiom, see Section~\ref{sec:comp}. We also refer to \cite{Oe:freefermi} for more details on the meaning of this renormalization.

\subsection{Additional positive axioms}
\label{sec:plusposax}

We proceed to consider additional axioms to encode the positive structure. In anticipation of later use we restrict the collection of hypersurfaces on which a positive structure should exist. To be concrete, denote by $\sts_1^p\subseteq \sts_1$ a subcollection of hypersurfaces that we shall refer to as \emph{polarized hypersurfaces}. This subcollection satisfies certain closedness conditions: $\sts_1^p$ is closed under orientation reversal and if $\Sigma\in\sts_1$ has a decomposition whose components are all in $\sts^p_1$ then $\Sigma$ itself is in $\sts^p_1$.

\begin{itemize}
\item[(R1+)] The real Krein space $\cDr_{\Sigma}$ associated to a polarized hypersurface $\Sigma$ carries a bigrading that combines to the f-grading. Moreover, the subspace $\cDr_{\Sigma,00}$ is an Archimedean ordered vector space with generating proper cone $\cDp_{\Sigma,00}$ such that
\begin{equation}
\lhs I^* \sigma',\sigma\rhs_{\Sigma} \ge 0\quad\text{if}\,\, \sigma,\sigma'\in\cDp_{\Sigma,00} .
\end{equation}
\item[(R2+)] Suppose the polarized hypersurface $\Sigma$ decomposes into a union of polarized hypersurfaces $\Sigma=\Sigma_1\cup\cdots\cup\Sigma_n$. Then, the associated map respects positivity in the sense,
\begin{equation}
 \tau^*_{\Sigma_1,\dots,\Sigma_n;\Sigma}(\sigma_1\tens\cdots\tens\sigma_n) \in \cDp_{\Sigma,00+}\quad\text{if}\,\,\forall i:\,\sigma_i\in\cDp_{\Sigma_i,00+} .
\end{equation}
\item[(R4+)] Suppose the boundary $\partial M$ of the region $M$ is a polarized hypersurface. Then, the probability map $A_M$ is positive on $\cDpd_{\partial M,00}$, i.e., $A_M(\sigma)\ge 0$ if $\sigma\in \cDpd_{\partial M,00}$. Moreover, $A_M$ vanishes on $\cDrd_{\partial M,11}$.
\end{itemize}

%% file: fockfeyn.tex
\section{Fock-Feynman quantization}
\label{sec:fockfeyn}

We present in this section a quantization scheme that produces from a fermionic classical free field theory in terms of the axioms of Appendix~\ref{sec:pclassax} a general boundary quantum field theory in terms of the real axioms of Section~\ref{sec:realax}. Since the quantization scheme is ultimately based on Fock spaces and the Feynman path integral, we refer to it as a \emph{Fock-Feynman quantization}. For short we refer to this scheme here as the FFR scheme. It is quite similar to the one presented in \cite{Oe:freefermi}, which however targets the amplitude formalism. We refer to the latter here as the FFA scheme. The FFR scheme is indeed designed to emulate the FFA scheme when the latter is combined with the step of taking the ``modulus square'' as described in Section~\ref{sec:twistedpf}. We discuss this further in Section~\ref{sec:addcomplex}. We heavily rely on notation and results from \cite{Oe:freefermi} in this section.

\subsection{Spaces on hypersurfaces}

Given a hypersurface $\Sigma$, axiom (C1) provides us with a real Krein space $L_{\Sigma}$. We define the real Krein space $\cDr_{\Sigma}$ as the Fock space over $L_{\Sigma}$. To emphasize the fact that we are dealing with a \emph{real} Fock space, we use here the notation $\cFr$. (In \cite{Oe:freefermi} the notation $\cF$ was used for both real and complex Fock spaces, although mostly complex Fock spaces were used.)

Concretely, $\cFr(L_{\Sigma})$ is the completion of an $\N_0$-graded Krein space,
\begin{equation}
\cFr(L_{\Sigma})\defeq \widehat{\bigoplus_{n=0}^{\infty}}\, \cFr_n(L_{\Sigma}) .
\end{equation}
Here, $\cFr_n(L_{\Sigma})$ denotes the real vector space of continuous $n$-linear maps from $n$ copies of $L_{\Sigma}$ to $\R$, that are completely anti-symmetric. The f-grading on $\cFr_n(L_{\Sigma})$ is defined to be even if $n$ is even and odd if $n$ is odd.
The inner product $g_n$ on $\cFr_n(L_{\Sigma})$ induced from the inner product on $L_{\Sigma}$ can be presented as follows. Choose an ON-basis $\{\zeta_i\}_{i\in I}$ of $L_{\Sigma}$.\footnote{Recall that an ON-basis on a Krein space $V$ is defined to be an ON-basis on $V_+$ combined with an ON-basis on $V_-$ with minus its inner product.} Then,
\begin{equation}
g_n(\sigma',\sigma)\defeq n! \sum_{i_1,\dots,i_n\in I} (-1)^{\sig{\zeta_1}+\dots+\sig{\zeta_n}}
 \sigma'(\zeta_1,\dots,\zeta_n) \sigma(\zeta_1,\dots,\zeta_n) .
\end{equation}
Because of the continuity of $\sigma$ and $\sigma'$ this is well defined and independent of the ON-basis chosen. The inner product $g_{\Sigma}$ on $\cFr(L_{\Sigma})$ arises from the sum of the inner products $g_n$ and subsequent completion in the Krein space sense.

The identification of the spaces $L_{\Sigma}$ and $L_{\overline{\Sigma}}$ via Axiom (C1) induces an identification of the Fock spaces $\cDr_{\Sigma}$ and $\cDr_{\overline{\Sigma}}$. Instead of using this directly, we insert a twist reflecting the change of orientation of the hypersurface in a reversal of the order of arguments. Concretely, an element $\sigma_{\Sigma}$ of degree $n$ in $\cFr_{\Sigma}$ is identified with the element $\sigma_{\overline{\Sigma}}$ in $\cFr_{\overline{\Sigma}}$ as follows,
\begin{equation}
\sigma_{\overline{\Sigma}}(\xi_1,\dots,\xi_n)= \sigma_{\Sigma}(\xi_n,\dots,\xi_1) .
\end{equation}
Axiom (C2) relates the inner products on $L_{\Sigma}$ and $L_{\overline{\Sigma}}$. This induces a relation between the inner products of $\cDr_{\Sigma}$ and $\cDr_{\overline{\Sigma}}$ as,
\begin{equation}
\lhs \sigma',\sigma\rhs_{\Sigma}= (-1)^{\fdg{\sigma}\cdot\fdg{\sigma'}} \lhs \sigma',\sigma\rhs_{\overline{\Sigma}} .
\end{equation}
Thus, Axiom (R1) is satisfied.

Given a decomposition $\Sigma_1\cup\Sigma_2$ of the hypersurface $\Sigma$ we define $\tau^*_{\Sigma_1,\Sigma_2;\Sigma}:\cDr_{\sigma_1}\ctens\cDr_{\Sigma_2}\to\cDr_{\Sigma}$ as follows. Given $\sigma_1\in\cDr_{\sigma_1}$ of degree $m$ and $\sigma_2\in\cDr_{\sigma_2}$ of degree $n$, using Axiom (C3) we set
\begin{multline}
\tau^*_{\Sigma_1,\Sigma_2;\Sigma}(\sigma_1\tens\sigma_2)((\eta_1,\xi_1),\dots (\eta_{m+n},\xi_{m+n})) \\
 \defeq \frac{1}{(m+n)!}\sum_{\tau\in S^{m+n}} (-1)^{|\tau|} \sigma_1(\eta_{\tau(1)},\dots,\eta_{\tau(m)})\sigma_2(\xi_{\tau(m+1)},\dots,\xi_{\tau(m+n)}) .
\end{multline}
This is extended by associativity to arbitrary decompositions. It is straightforward to check that this satisfies Axiom (R2).

\subsection{Probability maps}

Given a region $M$, Axioms (C4) and (C5) yield a real vector space $L_M$ and a linear map $r_M:L_M\to L_{\partial M}$ so that the image is a hypermaximal neutral subspace of $L_{\partial M}$. This data determines uniquely a linear map $u_M:L_{\partial M}\to L_{\partial M}$ with the following properties \cite[Lemma~3.1]{Oe:freefermi}: (a) $u_M$ is involutive, i.e., it squares to the identity. (b) $u_M$ is an adapted anti-isometry, i.e., $u_M$ interchanges the positive and negative parts of $L_{\partial M}$ anti-isometrically. (c) $u_M$ is the identity on $L_{\tilde{M}}$.

We define the amplitude map $A_M:\cDrd_{\partial M}\to\R$ as follows. Given an ON-basis $\{\zeta_i\}_{i\in I}$ of $L_{\partial M}$ and $\sigma\in\cDr_{\partial M}$ of degree $n$ we set,
\begin{multline}
 A_M(\sigma)\defeq \frac{(2n)!}{2^n n!} (-1)^n\\
 \sum_{i_1,\dots,i_n\in I}  (-1)^{\sig{\zeta_{i_1}}+\dots+\sig{\zeta_{i_n}}}
 \sigma(\zeta_{i_1},\dots,\zeta_{i_n},u_M\zeta_{i_n},\dots,u_M\zeta_{i_1}) .
\end{multline}
Note that this expression is not necessarily well defined. We define $\cDrd_{\partial M}$ to be the subspace of $\cDr_{\partial M}$ where linear combinations of this expression are well defined (i.e., the sums converge) and are independent of the choice of basis. This subspace is dense.\footnote{We leave out the demonstration  here which can be done by using ``generating states'', in complete analogy to what was done in \cite{Oe:freefermi}.} Since $u_{\overline{M}}=u_M$ it is straightforward to verify $A_M(\sigma)=A_{\overline{M}}(\sigma)$. Thus, we satisfy Axiom (R4).

Axioms (R3x) and (R5a) are also satisfied. We leave the verification as an exercise to the reader.

\subsection{Composition}
\label{sec:comp}

In contrast to the other axioms, the demonstration of Axiom~(R5b) or (R5b*) concerning the gluing of regions along hypersurfaces is rather non-trivial. Since this has been discussed at length in \cite{Oe:freefermi} for the FFA scheme, we will be very brief here and assume familiarity of the reader with the respective part of that paper.

Given a hypersurface $\Sigma$ we consider the direct system $\{L_{\Sigma,\alpha}\}_{\alpha\in A_{\Sigma}}$ of all finite dimensional Krein subspaces of $L_{\Sigma}$. Denote by $Q_{\Sigma,\alpha}$ the orthogonal projector onto $L_{\Sigma,\alpha}$. We induce from this a direct system $\{\cDr_{\Sigma,\alpha}\}_{\alpha\in A_{\Sigma}}$ of finite dimensional Krein subspaces of $\cDr_{\Sigma}$ as follows. Let $\cDr_{\Sigma,\alpha}$ be the subspace of $\cDr_{\Sigma}$ of those elements $\sigma$ that satisfy $\sigma=\sigma\circ Q_{\Sigma,\alpha}^{\tens n}$, where $n$ is the degree of $\sigma$. Note that $\cDr_{\Sigma,\alpha}$ is dual to the finite dimensional Fock space $\cF(L_{\Sigma,\alpha})$ which is a quotient of $\cDr_{\Sigma}$, but is also canonically isomorphic to this Fock space via the inner product. In particular, if $L_{\Sigma,\alpha}$ has dimension $d$, then $\cDr_{\Sigma,\alpha}$ has dimension $2^d$. We denote the orthogonal projector from $\cDr_{\Sigma}$ to $\cDr_{\Sigma,\alpha}$ by $P_{\Sigma,\alpha}$.
We recall what we mean by a direct limit of a collection $\{k_{\alpha}\}_{\alpha\in A_{\Sigma}}$ of real numbers indexed by $A_{\Sigma}$. We say that $k=\varinjlim_{\alpha} k_{\alpha}$ if for any $\epsilon >0$ there exists $\beta\in A_{\Sigma}$ such that,
\begin{equation}
 | k - k_{\gamma}|< \epsilon\quad \forall \gamma\ge\beta .
\end{equation}

Assume the geometric context of Axiom~(C7) or (R5b) or (R5b*). For any $\alpha\in A_{\Sigma}$, define the real number $\hat{c}_{\alpha}(M;\Sigma,\overline{\Sigma'})$, abbreviated as $\hat{c}_{\alpha}$, as follows,
\begin{equation}
\hat{c}_{\alpha} \defeq \sum_{i \in I} (-1)^{\sig{\xi_i}} A_M\left(\tau^*_{\Sigma_1,\Sigma,\overline{\Sigma'};\partial M}(\one\tens P_{\Sigma,\alpha}\xi_i\tens P_{\overline{\Sigma},\alpha}\xi_i)\right),
\end{equation}
where $\{\xi_i\}_{i\in I}$ is an arbitrary ON-basis of $\cDr_{\Sigma}$.

The key technical result in \cite{Oe:freefermi} underlying the proof of Axiom (R5b*) or (R5b) was Lemma~9.3 of that paper. We shall present a result with an analogous scope for the present context. Axiom~(C7) implies that given $\phi\in L_{\tilde{M}_1}$ there exists $\tilde{\phi}\in L_{\Sigma}$ such that $(\phi,\tilde{\phi},\tilde{\phi})\in L_{\tilde{M}}$. This was exploited crucially in Lemma~9.3 of \cite{Oe:freefermi}. Here we need a slightly more general statement as follows.
\begin{lem}
\label{lem:gluesolv}
Let $\phi\in L_{\Sigma_1}$. Then there exist $\phi',\phi''\in L_{\Sigma}$ such that,
\begin{equation}
 u_M(\phi,\phi',\phi'') = (u_{M_1}(\phi),\phi'',\phi') .
\label{eq:gluesolv}
\end{equation}
\end{lem}
Again the proof, which we leave to the reader, exploits principally Axiom~(C7). We also need the following Lemma which can be proven simply by choosing a basis.
\begin{lem}
\label{lem:solvgen}
Let $\alpha_1\in A_{\Sigma_1}$. Then there exists $\alpha\in A_{\Sigma}$ such that for all $\phi\in L_{\Sigma_1,\alpha_1}$ there are $\phi',\phi''\in L_{\Sigma,\alpha}$ satisfying Equation~(\ref{eq:gluesolv}) of Lemma~\ref{lem:gluesolv}.
\end{lem}

We are now ready to present the main result for the present quantization scheme, the demonstration that (R5b*) holds. This corresponds essentially to Theorem~10.2 in \cite{Oe:freefermi} and incorporates the analogue of Lemma~9.3 of that paper. We note that the elements $\sigma\in\cDr_{\Sigma_1}$ with the property that $\sigma\in\cDr_{\Sigma_1,\alpha_1}$ for some $\alpha_1\in A_{\Sigma_1}$ form a dense subspace of $\cDr_{\Sigma_1}$.
\begin{thm}
\label{thm:composition}
Let $\alpha_1\in A_{\Sigma_1}$ and $\sigma\in\cDr_{\Sigma_1,\alpha_1}$. Then, there exists $\alpha\in A_{\Sigma}$ such that for all $\gamma\ge\alpha$ and all ON-basis $\{\xi_{i}\}_{i\in I}$ of $L_{\Sigma}$ the following identity holds.
\begin{equation}
 A_{M_1}(\sigma)\cdot \hat{c}_{\gamma}(M;\Sigma,\overline{\Sigma'})
 = \sum_{i\in I} (-1)^{\sig{\xi_i}} A_M\left(\tau^*_{\Sigma_1,\Sigma,\overline{\Sigma'};\partial M}(\sigma\tens P_{\gamma}\xi_i\tens P_{\gamma}\xi_i)\right) .
\label{eq:fincompid}
\end{equation}
In particular, for the dense subspace of elements $\sigma\in\cDr_{\Sigma_1}$ of this form, the renormalized composition identity (\ref{eq:rencompid}) and thus Axiom~(R5b*) holds.
\end{thm}
\begin{proof} (Sketch)
We set $\alpha\in A_{\Sigma}$ according to Lemma~\ref{lem:solvgen}. It is clear that the relevant property is conserved for all $\gamma\ge\alpha$. The proof of the identity (\ref{eq:fincompid}) can then be constructed by following almost exactly along the lines of the proof of Lemma~9.3 of \cite{Oe:freefermi} in Appendix~A of that paper. Lemmas~A.1 and A.2 of that paper are for this purpose to be replaced by the following lemmas, (using the notation introduced in Lemma~\ref{lem:gluesolv}).
\begin{lem}
Let $\xi\in L_{\Sigma}$ and $\phi\in L_{\Sigma_1}$. Then,
\begin{align}
 g_{\partial M}\left((0,\xi,0), u_M(\phi,0,0)\right)
 & = g_{\Sigma}(\xi,\phi'') - g_{\partial M}\left((0,\xi,0),u_M(0,\phi',\phi'')\right), \\
 g_{\partial M}\left((0,0,\xi), u_M(\phi,0,0)\right)
 & = g_{\overline{\Sigma}}(\xi,\phi') - g_{\partial M}\left((0,0,\xi),u_M(0,\phi',\phi'')\right) .
\end{align}
\end{lem}
\begin{lem}
Let $\phi_1,\phi_2\in L_{\Sigma_1}$. Then,
\begin{multline}
 g_{\Sigma_1}(\phi_1,u_{M_1}(\phi_2))
 = g_{\partial M}\left((\phi_1,0,0),u_M(\phi_2,0,0)\right) \\
 - g_{\partial M}\left((0,\phi_1',\phi_1''),u_M(0,\phi_2',\phi_2'')\right)
 + g_{\Sigma}(\phi_1',\phi_2'') + g_{\overline{\Sigma}}(\phi_1'',\phi_2') .
\end{multline}
\end{lem}
This concludes the sketch of the proof of Theorem~\ref{thm:composition}.
\end{proof}

In the special case that for a given classical theory the direct limits $\varinjlim_{\alpha}\hat{c}_\alpha$ exist for all admissible gluings of regions along hypersurfaces we say that this theory satisfies the \emph{integrability condition}, compare \cite{Oe:holomorphic,Oe:freefermi}. In this case we have validity of the slightly stronger Axiom~(R5b) with the gluing anomaly factors given by these limits, $\hat{c}\defeq\varinjlim_{\alpha}\hat{c}_\alpha$. This result is the analogue of Theorem~9.2 in \cite{Oe:freefermi}.

%% file: addcomplex.tex
\section{Adding complex structure}
\label{sec:addcomplex}

In this section we show that the addition of complex structure to the classical data leads to the quantum theory satisfying in addition the positivity axioms of Section~\ref{sec:plusposax}. In total, the quantum theory then satisfies all the axioms of the positive formalism.
With the complex structure present one may alternatively quantize the classical theory with the FFA scheme introduced in \cite{Oe:freefermi}, outputting to the amplitude formalism. What we also show in this section is that the resulting quantizations exactly coincide, when the FFA scheme is combined with the ``modulus square'' operation detailed in Section~\ref{sec:twistedpf}.

Crucially, we have the option to add complex structure only on some hypersurfaces. To this end we consider a subcollection of \emph{polarized hypersurfaces} (see beginning of Section~\ref{sec:plusposax}) for which the data of the classical theory is augmented by complex structures according to the axioms of Appendix~\ref{sec:sclassax}. The following considerations are thus limited to these hypersurfaces and regions bounded by such hypersurfaces.

A complex structure on a real vector space allows to decompose functions on this vector space into holomorphic and anti-holomorphic parts. This applies in particular to elements of a Fock space over a real vector space since we have modeled them as functions. Given a polarized hypersurface $\Sigma$, we recall that the associated space $\cDr_{\Sigma}$ is the Fock space $\cFr(L_{\Sigma})$ over the real Krein space $L_{\Sigma}$. We decompose an element $\sigma\in\cFr_n(L_{\Sigma})$ uniquely as,
\begin{equation}
 \sigma(\xi_1,\dots,\xi_{n})
= \sum_{k=0}^{n} \frac{1}{n!}\sum_{\sigma\in S^{n}}(-1)^{|\sigma|}
 \psi^{\hol}_k(\xi_{\sigma(1)},\dots,\xi_{\sigma(k)})\psi^{\ahol}_{n-k}(\xi_{\sigma(n)},\dots,\xi_{\sigma(k+1)}) .
\label{eq:dechol}
\end{equation}
Here the $\psi_k^{\hol}$ are holomorphic functions while the $\psi_k^{\ahol}$ are anti-holomorphic functions. Both types are required to be completely anti-symmetric.

We may now take the functions $\psi_k^{\hol}$ to be elements of the Fock space $\cF(L_{\Sigma})$ over $L_{\Sigma}$ understood as a \emph{complex} vector space. Here $L_{\Sigma}$ is equipped with the complex inner product given by (\ref{eq:complip}) and (\ref{eq:fsympl}) which induces the inner product of $\cF(L_{\Sigma})$. We follow the conventions of \cite{Oe:freefermi}. Given an arbitrary ON-basis $\{\zeta_j\}_{j\in K}$ of the \emph{real} Krein space $L_{\Sigma}$ the inner product on $\cF_n(L_{\Sigma})$ can be expressed as,
\begin{equation}
\langle \eta,\psi\rangle_{\Sigma}
= n! \sum_{j_1,\dots,j_n\in K} (-1)^{\sig{\zeta_{j_1}}+\dots+\sig{\zeta_{j_n}}}
 \overline{\eta(\zeta_{j_1},\dots,\zeta_{j_n})}
 \psi(\zeta_{j_1},\dots,\zeta_{j_n}) .
\end{equation}
This is equivalent to formula (3.2) given in \cite{Oe:freefermi}.
Similarly, we may take the functions $\psi_k^{\ahol}$ to be elements of the Fock space $\cF(L_{\overline{\Sigma}})$.

Recall that the state space $\cH_{\Sigma}$ associated to the polarized hypersurface $\Sigma$ in the FFA scheme is precisely the Fock space $\cF(L_{\Sigma})$. Take $\cD_{\Sigma}$ to denote the complexification of $\cDr_{\Sigma}$ (i.e., extending real valued to complex valued functions). Then it is easy to verify that the map $\cD_{\Sigma}\to\cH_{\Sigma}\tens\cH_{\overline{\Sigma}}$ induced by the decomposition (\ref{eq:dechol}) relates the inner products of these spaces precisely as determined by formula (\ref{eq:iptwist}). That is, we have successfully identified on the level of polarized hypersurfaces the (complexified) output of the FFR scheme with the output of the FFA scheme with subsequent ``modulus square'' operation.

It now remains to verify that given this identification all further structures associated to polarized hypersurfaces or regions bounded be such really coincide between the (complexified) output of the FFR scheme as described in Section~\ref{sec:fockfeyn} and the output of the ``modulus squared'' FFA scheme of \cite{Oe:freefermi} according to Section~\ref{sec:twistedpf}. We limit ourselves here to two remarks and leave this straightforward task otherwise to the reader.
Firstly, it is useful for this purpose to rewrite the amplitude map of the FFA scheme using an ON-basis $\{\zeta_j\}_{j\in K}$ of the \emph{real} Krein space $L_{\partial M}$ as follows,
\begin{multline}
 \rho_M(\psi)= \frac{(2n)!}{2^n n!} (-1)^n \\
 \sum_{j_1,\dots,j_n\in K}  (-1)^{\sig{\zeta_{j_1}}+\dots+\sig{\zeta_{j_n}}}
 \psi(\zeta_{j_1},\dots,\zeta_{j_n},u_M\zeta_{j_n},\dots,u_M\zeta_{j_1}) ,
\end{multline}
where $\psi\in\cH^{\ds}_{\partial M}$ is of degree $2n$. This is easily verified to be equivalent to formula (8.9) in \cite{Oe:freefermi}.
Secondly, the following observation is helpful. Given a continuous complex bilinear map $a:L_{\Sigma}\tens L_{\Sigma}\to\C$ we have,
\begin{equation}
\sum_{j\in K} (-1)^{\sig{\zeta_j}} a(\zeta_j,\zeta_j)=0 .
\end{equation}
Here $\{\zeta_j\}_{j\in K}$ is an ON-basis of the Krein space $L_{\Sigma}$ viewed either as a real or as a complex Krein space.

The validity of the additional positive axioms of Section~\ref{sec:plusposax} in the polarized sector of the FFA scheme can now be simply inferred from the validity of those axioms (as part of the positive formalism) in the ``modulus square'' of the FFR scheme \cite{Oe:freefermi}, due to the equivalence of the two.

%% file: outlook.tex
\section{Discussion and outlook}
\label{sec:outlook}

The standard formulation of quantum theory in terms of a single Hilbert space of states with observables as operators on it treats space and time in very different ways. In particular, it depends on a predetermined notion of time and it is not manifestly local in space, even though the physics seems to be. This is awkward from a special relativistic point of view and fatal from a general relativistic point of view. It motivates the search for a manifestly local formulation of fundamental quantum physics, compatible with general relativistic principles.

With the advent of topological quantum field theory (TQFT) due to E.~Witten, G.~Segal, M.~Atiyah and others at the end of the 1980s the right mathematical tools finally seemed to be at hand. Indeed, a mathematical revolution ensued, impacting the fields of low-dimensional topology, algebraic topology, knot theory and category theory among others. However, on the physics side applications have been limited mainly to 2-dimensional conformal field theory, and toy models such as Chern-Simons theory that possess finite-dimensional state spaces only.

With the general boundary formulation (GBF) the TQFT approach is pushed further to a full blown formulation of quantum theory. The aim of this program is to develop a more suitable basis of the foundations of quantum theory, more fundamental in fact than the standard formulation. The crucial test for the GBF, as for any theory of physics, is the agreement with experiment. In the present case, with the GBF being a framework rather than an actual model, this means that the known models successfully describing fundamental quantum physics must be shown to be compatible with the GBF. Concretely, quantum field theory (or a ``realistic'' subset thereof) must be shown to be compatible with the GBF.

A key obstacle for TQFT-type approaches in general and the GBF in particular to achieve a local description of realistic quantum field theories is precisely the state locality problem outlined in the Introduction. It should be emphasized that this is a basic issue, quite separate from other important challenges such as those arising from non-linearity, constraints or gauge symmetries. It is thus suitable to address this problem in a linear context, i.e., via free field theory, as we have done in this paper.

The axioms of Appendix~\ref{sec:pclassax} appear to provide a reasonable description of classical field theories underlying simple free bosonic and fermionic quantum field theories \cite{Oe:holomorphic,Oe:freefermi}. To be concrete and pick a physically relevant model consider the Dirac field theory on Minkowski spacetime. We have shown in \cite[Section~5.1]{Oe:freefermi} explicitly how a Krein space of suitable $L^2$-sections of the spinor bundle is associated to any smooth hypersurface in Minkowski spacetime.\footnote{The prescription in \cite[Section~5.1]{Oe:freefermi} restricts the admissible hypersurfaces to be such that their null subset has measure zero.} What is less obvious is that the space of solutions of the Dirac equation in a spacetime region should be a hypermaximal neutral subspace of this space of data on the boundary of the region. Simple Lagrangian arguments show that it should be neutral at least and additional hypermaximal neutrality arises then from a kind of non-degeneracy condition. We conjecture that this is satisfied for sufficiently regular spacetime regions such as for example cubes (with spacelike and timelike edges) or lenses (i.e., regions bounded by a pair of compact spacelike hypersurfaces joined along a null sphere).

In order to quantize the theory we need additional structure, however. Following traditional ideas of quantization (and geometric quantization in particular), this can be encoded in a complex structure on each of the spaces of hypersurface data (see the axioms listed in Appendix~\ref{sec:sclassax}). The scheme presented in \cite{Oe:freefermi} then yields the quantum theory in terms of the amplitude formalism which directly generalizes the concept of ``local Hilbert spaces''. The theory so obtained can be compared with the traditional form of the Dirac quantum field theory in the standard formulation and is indeed ``correct''. The problem with this is that the relevant ``correct'' complex structures only exist for a class of hypersurfaces without boundary (including all Cauchy hypersurfaces). This is because they arise from pseudodifferential and thus non-local operators rather than from differential operators.

Consequently, this quantization scheme and the obtained quantum theory are only defined for a very restricted class of regions and hypersurfaces. In particular, the regions in question do not include compact regions. So a truly local description of the physics is not obtained. This is the state locality problem in the present context.

The strategy we have followed in this paper to address the problem consists firstly in changing the formalism in which to express the quantum theory. Rather than the amplitude formalism we use the positive formalism introduced in \cite{Oe:dmf}. This is somewhat (but not completely) analogous to working with mixed states rather than pure states in the standard formulation. Since the positive formalism contains all the operationally relevant information about a quantum theory that the amplitude formalism contains we are not losing anything by this step. Secondly, we exhibit, for the fermionic case, a quantization scheme that outputs to the positive formalism. We show, moreover, that this quantization scheme yields equivalent output to the ``correct'' scheme of \cite{Oe:freefermi} when applied to a classical field theory with complex structures, i.e., one that satisfies the axioms of both Appendices~\ref{sec:pclassax} and \ref{sec:sclassax}. The comparison is made via a functor from the amplitude formalism to the positive formalism, slightly modified from the one presented in \cite{Oe:dmf}, see Section~\ref{sec:twistedpf}.

Thirdly, we have shown that this new scheme (presented in Section~\ref{sec:fockfeyn}) works even when the classical theory satisfies only the axioms of Appendix~\ref{sec:pclassax} and no complex structures are present. In this case, the output is restricted to a ``real formalism'' (see Section~\ref{sec:realax}), a weakened version of the positive formalism with positive structures reduced to real ones. The state locality problem is avoided and we obtain perfectly well defined local spaces of ``generalized mixed states'' associated to hypersurfaces with boundary in the quantum theory.

What is more, the classical theory can be augmented with complex structures for only a selected collection of hypersurfaces. (We have referred to these as \emph{polarized hypersurfaces}.) This then leads to a corresponding augmentation by positive structures on the quantum side (axioms of Section~\ref{sec:plusposax}), restoring the full positive formalism for these hypersurfaces.

For the probability interpretation of the GBF the positive structures are crucial \cite{Oe:dmf}. Thus, losing them presents a serious issue. Note, however, that probabilities can only be extracted for measurements in spacetime regions. So positive structures are only relevant on boundaries of regions. Without losing predictive power we can therefore discard on the classical side at least the complex structures associated to hypersurfaces that are not boundaries. This includes all hypersurfaces that have boundaries themselves.

It remains to address the lack of positive structures on those hypersurfaces that are boundaries. Of course we do have complex structures (and thus induced positive structures) on some of these hypersurfaces, as explained above. The minimalist approach would be to simply accept the loss of predictive power for the regions whose boundaries cannot be equipped with complex structures. (Note that this only affects predictive power not present in the standard formulation either). A more interesting approach would be to seek a new way to construct the positive structures, possibly using weaker or different additional classical data.

The quantization scheme presented in Section~\ref{sec:fockfeyn} is restricted to fermionic field theory. This restriction was not motivated by any mathematical difficulty, but by a physical one. It is perfectly possible to formulate the scheme also for bosonic theories (except for the renormalization of the gluing anomaly, compare \cite{Oe:freefermi}). The problem is that the basic structure required on the space of germs of solutions on a hypersurface $\Sigma$ is that of a non-degenerate symmetric bilinear form $g_{\Sigma}$, exactly as in the fermionic case. But while this is natural in the fermionic case, the structure naturally present in the bosonic case is rather a non-degenerate anti-symmetric bilinear form $\omega_{\Sigma}$, compare Appendix~\ref{sec:pclassax}. The symmetric bilinear form $g_{\Sigma}$ is only obtained from this with the help of the complex structure (via formula (\ref{eq:bsym})), whose use we are precisely trying to avoid.

An important topic that we have not touched upon in the present paper is that of observables, or more generally that of quantum operations. See \cite{Oe:dmf} for some remarks in that direction. A proper development of these concepts in the present context could in turn help to clarify the notion of positivity.

We might also ask what physically measurable objects are encoded by states on hypersurfaces with boundaries. This question is of particular interest when the hypersurface is not polarized, i.e., when a complex structure is not available. Since we have used Fock spaces it is tempting to interpret the states as particle states in the usual manner. However, in the real formalism with real Fock spaces this would certainly not be correct. Indeed, we recall formula (\ref{eq:dechol}) from Section~\ref{sec:addcomplex} that shows how a real Fock state decomposes into the usual complex Fock states once we augment the hypersurface to a polarized hypersurface, i.e., add complex structure. In particular, a real (mixed) state of degree $n$ decomposes into products of complex states of degree $k$ and $n-k$ for $0\le k\le n$. Apart from this, to know if a real (mixed) state is even positive and thus physical in the conventional sense requires the complex structure.
A more operational approach to this problem would be to study the response of such states to local observables or other models for local (particle) measurement devices. A promising line of inquiry in this direction is suggested by the work of Colosi and Rovelli who introduced a notion of local particle as determined by a localized Hamiltonian operator \cite{CoRo:particle}. A more precise implementation of this idea in bosonic \cite{VRWL:locquanta} and fermionic \cite{KuWeLe:locdirac} quantum field theory was achieved subsequently by León and collaborators, using representations (equivalently: complex structures) inequivalent to the standard one. Generalizing this to the present context would require at the very least a basic understanding of physical observables, see the previous remark.

The positive formalism and its reduced ``real'' version introduced here also invite the contemplation of approaches to generalized states that are rather different from the Fock space approach. In particular, the notion of state of algebraic quantum field theory (AQFT) suggests itself for adaptation to this context. Rather than associate algebras of observables to regions as in AQFT, such algebras would be associated to hypersurfaces. The noncommutative product of the algebra would then simply arise from the ordered gluing of these hypersurfaces to themselves, viewed as slice regions. The space of generalized states could then be defined as the functionals on the algebra of the respective hypersurface. What is less straightforward in this approach is the encoding of the dynamics, i.e., the probability maps. Nevertheless, this could be the starting point of a somewhat different attack on the state locality problem as compared to the one presented here.

%% file: sts.tex
\section{Spacetime systems}
\label{sec:sts}

We lay out here a formalization of the notion of spacetime in terms of a \emph{spacetime system} for the GBF. This formalization is similar to that in \cite{Oe:freefermi}, but with the crucial difference that we allow here hypersurfaces with boundaries and corresponding more general decompositions of these as in \cite{Oe:2dqym}.

There is a fixed positive integer $d\in\N$, the \emph{dimension} of spacetime. We are given a collection $\sts_0^{\textrm{c}}$ of connected oriented topological manifolds of dimension $d$, possibly with boundary, that we call connected \emph{regular regions}. Furthermore, there is a collection $\sts_1^{\mathrm{c}}$ of connected oriented topological manifolds of dimension $d-1$, possibly with boundary, that we call \emph{hypersurfaces}. The manifolds are either abstract manifolds or they are all concrete closed regular submanifolds of a given fixed \emph{spacetime manifold}. In the former case we call the spacetime system \emph{local}, in the latter we call it \emph{global}.

There is a notion of disjoint union both for regular regions and for hypersurfaces. This leads to the collection $\sts_0$, of all formal finite unions of elements of $\sts_0^{\mathrm{c}}$, and to the collection $\sts_1$, of all formal finite unions of elements of $\sts_1^{\mathrm{c}}$. In the local case, the unions may be realized concretely as disjoint unions. In the global case, only unions with members whose interiors are disjoint are allowed in $\sts_1$ and $\sts_0$. Note that in this case the elements of $\sts_1$ and $\sts_0$ cannot necessarily be identified with submanifolds of the spacetime manifold as overlaps on boundaries may occur.

For elements of $\sts_1$ there is a notion of \emph{decomposition}. Given a presentation of a hypersurface $\Sigma$ as the union of hypersurfaces $\Sigma_1,\dots,\Sigma_n$ we call this a decomposition if (a) each $\Sigma_i$ is closed in $\Sigma$ and (b) the intersection of $\Sigma_i$ with $\Sigma_j$ is contained in their boundaries for each $i$ and $j$ with $i\neq j$.

The collections $\sts_1$ and $\sts_0$ are closed under orientation reversal. Also, any boundary of a regular region in $\sts_0$ is a hypersurface in $\sts_1$. That is, taking the boundary defines a map $\partial:\sts_0\to\sts_1$. When we want to emphasize explicitly that a given manifold is in one of those collections we also use the attribute \emph{admissible}.

It is convenient to also introduce the concept of \emph{slice regions}.\footnote{In some previous papers slice regions were called ``empty regions''.} A slice region is topologically a hypersurface, but thought of as an infinitesimally thin region. Concretely, the slice region associated to a hypersurface $\Sigma$ will be denoted by $\hat{\Sigma}$ and its boundary is defined to decompose as the disjoint union $\partial \hat{\Sigma}=\Sigma\cup\overline{\Sigma}$. There is one slice region for each hypersurface (forgetting its orientation). We refer to regular regions and slice regions collectively as \emph{regions}.

There is also a notion of \emph{gluing} of regions. Suppose we are given a region $M$ with its boundary decomposing as the union $\partial M=\Sigma_1\cup\Sigma\cup\overline{\Sigma'}$, where $\Sigma'$ is a copy of $\Sigma$. ($\Sigma_1$ may be empty.) Then, we may obtain a new region $M_1$ by gluing $M$ to itself along $\Sigma,\overline{\Sigma'}$. That is, we identify the points of $\Sigma$ with corresponding points of $\Sigma'$ to obtain $M_1$. The resulting region $M_1$ might be inadmissible, in which case the gluing is not allowed.

Depending on the theory one wants to model, the manifolds may carry additional structure such as for example a differentiable structure or a metric. This has then to be taken into account in decompositions and gluings.

%% file: classax.tex
\section{Axioms of free classical field theory}
\label{sec:classax}

\subsection{Purely classical axioms}
\label{sec:pclassax}

Given a spacetime system, we axiomatize a linear classical field theory on the spacetime system as follows. This is a version of the axioms as put forward in \cite{Oe:freefermi}, but reduced to purely classical data, i.e., without complex structures. Also, the requirement for hypersurface decompositions to be disjoint is dropped.

\begin{itemize}
\item[(C1)] Associated to each hypersurface $\Sigma$ is a real vector space $L_{\Sigma}$. In the fermionic case $L_{\Sigma}$ is equipped with a real inner product $g_{\Sigma}$, making it into a separable Krein space. In the bosonic case $L_{\Sigma}$ is equipped with a non-degenerate anti-symmetric bilinear form $\omega_{\Sigma}$, making it into a symplectic vector space.
\item[(C2)] Associated to each hypersurface $\Sigma$ there is an (implicit) linear involution $L_\Sigma\to L_{\overline\Sigma}$, such that $g_{\overline{\Sigma}}=-g_{\Sigma}$ in the fermionic case and $\omega_{\overline{\Sigma}}=-\omega_{\Sigma}$ in the bosonic case.
\item[(C3)] Suppose the hypersurface $\Sigma$ decomposes into a union of hypersurfaces $\Sigma=\Sigma_1\cup\cdots\cup\Sigma_n$. Then, there is an (implicit) isomorphism $L_{\Sigma_1}\oplus\cdots\oplus L_{\Sigma_n}\to L_\Sigma$. The isomorphism preserves the inner product or the symplectic form.
\item[(C4)] Associated to each region $M$ is a real vector space $L_M$.
\item[(C5)] Associated to each region $M$ there is a linear map of real vector spaces $r_M:L_M\to L_{\partial M}$. In the fermionic case the image $L_{\tilde{M}}$ of $r_M$ is a real hypermaximal neutral subspace of $L_{\partial M}$. In the bosonic case the image $L_{\tilde{M}}$ of $r_M$ is a Lagrangian subspace of $L_{\partial M}$.
\item[(C6)] Let $M_1$ and $M_2$ be regions and $M= M_1\cup M_2$ be their disjoint union. Then $L_M$ is the direct sum $L_{M}=L_{M_1}\oplus L_{M_2}$. Moreover, $r_M=r_{M_1}+r_{M_2}$.
\item[(C7)] Let $M$ be a region with its boundary decomposing as a union $\partial M=\Sigma_1\cup\Sigma\cup \overline{\Sigma'}$, where $\Sigma'$ is a copy of $\Sigma$. Let $M_1$ denote the gluing of $M$ to itself along $\Sigma,\overline{\Sigma'}$ and suppose that $M_1$ is a region. Then, there is an injective linear map $r_{M;\Sigma,\overline{\Sigma'}}:L_{M_1}\toi L_{M}$ such that
\begin{equation}
 L_{M_1}\toi L_{M}\rightrightarrows L_\Sigma
\label{eq:xsbdy}
\end{equation}
is an exact sequence. Here the arrows on the right hand side are compositions of the map $r_M$ with the projections of $L_{\partial M}$ to $L_\Sigma$ and $L_{\overline{\Sigma'}}$ respectively (the latter identified with $L_\Sigma$). Moreover, the following diagram commutes, where the bottom arrow is the projection.
\begin{equation}
\xymatrix{
  L_{M_1} \ar[rr]^{r_{M;\Sigma,\overline{\Sigma'}}} \ar[d]_{r_{M_1}} & & L_{M} \ar[d]^{r_{M}}\\
  L_{\partial M_1}  & & L_{\partial M} \ar[ll]}
\end{equation}
\end{itemize}

\subsection{Adding complex structure}
\label{sec:sclassax}

We provide here a separate list of axioms encoding additional semiclassical information in the form of complex structures on hypersurfaces. In order to allow for the possibility that this additional structure might not be present on all hypersurfaces we employ the notion of polarized hypersurface introduced in Section~\ref{sec:plusposax}. If all hypersurfaces are polarized the axioms together with those of Appendix~\ref{sec:pclassax} are equivalent to those of Section~4.2 of \cite{Oe:freefermi}.

\begin{itemize}
\item[(C1+)] For each polarized hypersurface $\Sigma$, the vector space $L_{\Sigma}$ is equipped with a complex structure $J_{\Sigma}:L_{\Sigma}\to L_{\Sigma}$. In the fermionic case the complex structure preserves the Krein space decomposition and is an isometry. In the bosonic case the complex structure preserves the symplectic structure and leads to a complete inner product.
\item[(C2+)] The complex structure changes sign under change of orientation. That is, $J_{\overline{\Sigma}}=-J_{\Sigma}$.
\item[(C3+)] Suppose a polarized hypersurface $\Sigma$ decomposes into a union of polarized hypersurfaces $\Sigma=\Sigma_1\cup\cdots\cup\Sigma_n$. Then, the associated complex structures are related via, $J_{\Sigma}= J_{\Sigma_1}+\cdots +J_{\Sigma_n}$.
\item[(C5+)] Suppose the boundary $\partial M$ of the region $M$ is a polarized hypersurface. In the fermionic case we require $J_{\partial M}\circ u_M=-u_M\circ J_{\partial M}$.
\end{itemize}
Note that the map $u_M$ in (C5+) is the the map determined from the hypermaximal neutral subspace $L_{\tilde{M}}\subseteq L_{M}$ via Lemma~3.1 of \cite{Oe:freefermi}.